\documentclass[a4paper,12pt]{article}
\usepackage{amsfonts}
\usepackage{fullpage}
\usepackage{subfigure}
\usepackage{amsmath}
\usepackage{amsthm}
\usepackage{comment}
\usepackage{graphicx}
\usepackage{amssymb,epsfig}
\usepackage[usenames,dvipsnames]{xcolor}
\usepackage{subfigure}
\usepackage[authoryear]{natbib}
\usepackage{amssymb}
\usepackage{algorithm}
\usepackage{multirow}

\newtheorem{theorem}{Theorem}[section]
\newtheorem{lemma}[theorem]{Lemma}

\newcommand{\Expect}[1]{\mathbb{E} \left[{#1}\right]}
\newcommand{\Expects}[2]{\mathbb{E}_{{#1}} \left[{#2}\right]}

\newcommand\blfootnote[1]{%
  \begingroup
  \renewcommand\thefootnote{}\footnotetext{#1}%
  \addtocounter{footnote}{-1}%
  \endgroup
}

\graphicspath{{./figures/}}

\begin{document}

\title{Particle Approximations of the Score and Observed Information Matrix for Parameter Estimation in State Space Models With Linear Computational Cost}
\author{Christopher Nemeth, Paul Fearnhead and Lyudmila Mihaylova}
\blfootnote{Christopher Nemeth, Department of Mathematics and Statistics, Lancaster University, Lancaster LA1 4YF, UK (Email: c.nemeth@lancaster.ac.uk). Paul Fearnhead, Department of Mathematics and Statistics, Lancaster University, Lancaster LA1 4YF, UK (Email: p.fearnhead@lancaster.ac.uk). Lyudmila Mihaylova, Department of Automatic Control and Systems Engineering, University of Sheffield, Sheffield S1 3JD, UK (Email: l.s.mihaylova@sheffield.ac.uk).}

\maketitle

\begin{abstract}
\cite{Poyiadjis2011} show how particle methods can be used to estimate both the score and the observed information matrix for state space models. These methods either suffer from a computational cost that is quadratic in the number of particles, or produce estimates whose variance increases quadratically with the amount of data. This paper introduces an alternative approach for estimating these terms at a computational cost that is linear in the number of particles. The method is derived using a combination of kernel density estimation, to avoid the particle degeneracy that causes the quadratically increasing variance, and Rao-Blackwellisation. Crucially, we show the method is robust to the choice of bandwidth within the kernel density estimation, as it has good asymptotic properties regardless of this choice. Our estimates of the score and observed information matrix can be used within both online and batch procedures for estimating parameters for state space models. Empirical results show improved parameter estimates compared to existing methods at a significantly reduced computational cost. Supplementary materials including code are available.
\end{abstract}


% \end{keywords}

\section{Introduction}
\label{sec:introduction}

State space models have become a popular framework to model nonlinear time series problems in engineering, econometrics and statistics \cite[]{cappe2005inference,durbin2001time}.
In this paper we consider the problem of maximum likelihood estimation of the model parameters, $\theta$, for nonlinear, non-Gaussian state space models, where there is no closed form expression for the marginal likelihood, $p(y_{1:T}|\theta)$, for data $y_{1:T}=\{y_1,y_2,\ldots,y_T\}$.

Using sequential Monte Carlo (SMC) methods, also known as particle filters, we propose an efficient method to create particle approximations of the score vector $\nabla \log p(y_{1:T}|\theta)$, which can be used within a gradient ascent algorithm to estimate $\theta$ by indirectly maximising the likelihood function. We show that our proposed algorithm can be applied offline, to estimate the $\theta$ from batches of data, or recursively, to update $\theta$ when new observations $y_t$ are received. Previous work by \cite{Poyiadjis2011}, has provided two approaches for estimating the score vector and observed information matrix. The first has a computational complexity that is linear in the number of particles, but it has the drawback that the variance of the estimates increases quadratically through time. The second method produces estimates whose variance increases linearly with time, but at the expense of a computational cost that is quadratic in the number of particles. The increased computational complexity of this algorithm limits its use for online applications.  

We propose a new method for estimating the score vector and observed information matrix using a novel implementation of a kernel density estimation technique \citep{West}, with Rao-Blackwellisation to reduce the Monte Carlo error of our estimates. The result is a linear-time algorithm which has substantially smaller Monte Carlo variance than the linear-time algorithm of \cite{Poyiadjis2011} and notable improvements over the fixed-lag smoother \citep{Olsson2008a} -- with empirical results showing the Monte Carlo variance of the estimate of the score vector increases only linearly with time. Furthermore, unlike standard uses of kernel density estimation, we derive results showing that our method is robust to the choice of bandwidth. For any fixed bandwidth our approach can consistently estimate the parameters as both the number of time-points and the number of particles go to infinity. 

Our final algorithm has similarities with the fixed-lag smoother of \cite{Dahlin2013c}, in terms of reducing the Monte Carlo error in the score and observed information estimates. However, one of the key advantages of our approach using Rao-Blackwellisation and kernel density estimation is that we are able to better approximate the observed information matrix, which in turn leads to faster and more accurate parameter estimation. A recently proposed linear time algorithm by \cite{westerborn2014}, supported by theoretical results \citep{Olsson2014}, could be also be used, but is not tested here. Finally, compared to competing methods, empirical results on a challenging eight parameter nonlinear model show that our algorithm produces more consistent parameter estimates, with an order of magnitude improvement in the rate of convergence.

\section{Inference for state space models}
\label{sec:particle-filtering}

\subsection{State space models}
\label{sec:state-space-models}

Consider the general state space model where $\{X_t; 1 \leq t \leq T\}$ represents a latent Markov process that takes values on $\mathcal{X} \subseteq \mathbb{R}^{n_x}$. The process is fully characterised by its initial density $p(x_1|\theta) = \mu_\theta(x_1)$ and transition probability density
\begin{equation}
  \label{eq:2}
  p(x_t | x_{1:t-1},\theta) = p(x_t|x_{t-1},\theta) = f_\theta(x_t|x_{t-1}),
\end{equation}
where $\theta \in \Theta$ represents a vector of model parameters. For an arbitrary sequence $\{z_i\}$ the notation $z_{i:j}$ corresponds to $(z_i,z_{i+1},\ldots,z_j)$ for $i\leq j$.

We assume that the process $\{X_t\}$ is not directly observable, but partial observations can be made via a second process $\{Y_t; 1 \leq t \leq T\} \subseteq \mathcal{Y} \subseteq \mathbb{R}^{n_y}$. The observations $\{Y_t\}$ are conditionally independent given $\{X_t\}$ and are defined by the probability density
\begin{equation}
  \label{eq:3}
  p(y_t|y_{1:t-1},x_{1:t},\theta) = p(y_t|x_t,\theta) = g_\theta(y_t|x_t).
\end{equation}

In the standard Bayesian context the latent process $\{X_{1:T}\}$ is estimated conditional on a sequence of observations $y_{1:T}$, for $T \geq 1$. If the parameter vector $\theta$ is known then the conditional distribution $p(x_{1:T}|y_{1:T},\theta) \propto p(x_{1:T},y_{1:T},\theta)$ can be evaluated where
\begin{equation}
  \label{eq:4}
  p(x_{1:T},y_{1:T},\theta) = \mu_\theta(x_1) \prod_{t=2}^{T}f_\theta(x_t|x_{t-1}) \prod_{t=1}^{T}g_\theta(y_t|x_t).
\end{equation}

For nonlinear, non-Gaussian state space models it is not possible to evaluate the posterior density $p(\theta,x_{1:T}|y_{1:T})$ in closed form. A popular approach for approximating these densities is to use a sequential Monte Carlo algorithm. 

\subsection{Sequential Monte Carlo algorithm}
\label{sec:sequ-monte-carlo}

SMC algorithms allow for the sequential approximation of the conditional density of the latent state given a sequence of observations, $y_{1:t}$, for a fixed $\theta$, which in this section we assume are known model parameters. For simplicity we shall focus on methods aimed at approximating the conditional density for the current state, $X_t$, but the ideas can be extended to learning about the full path of the process, $X_{1:t}$. Approximations of the density $p(x_{t}|y_{1:t},\theta)$ can be calculated recursively by first approximating $p(x_1|y_1,\theta)$, then $p(x_{2}|y_{1:2},\theta)$ and so forth. Each conditional density can be approximated by a set of $N$ weighted random samples, called particles, where
\begin{equation}
\label{eq:22}
  \hat{p}(dx_{t}|y_{1:t},\theta) = \sum_{i=1}^N w_{t}^{(i)}\delta_{X^{(i)}_{t}}(dx_{t}), \quad \forall i \: w_{t}^{(i)}\geq 0, \quad \sum_{i=1}^N w_{t}^{(i)}=1
\end{equation}
is an approximation for the conditional distribution and $\delta_{x_{0}}(dx)$ is a Dirac delta mass function located at $x_0$. The set of particles $\{X_{t}^{(i)}\}_{i=1}^N$ and their corresponding weights $\{w_t^{(i)}\}_{i=1}^{N}$ provide an empirical measure that approximates the probability density function $p(x_{t}|y_{1:t},\theta)$, where the accuracy of the approximation increases as $N \rightarrow \infty$ \citep{Crisan2002}.

We can recursively update our approximation using the following filtering recursion,
\begin{equation}
  \label{eq:6}
  p(x_t|y_{1:t},\theta) \propto g_\theta(y_t|x_t) \int f_\theta(x_t|x_{t-1})p(x_{t-1}|y_{1:t-1},\theta) dx_{t-1},
\end{equation}
where if we assume that at time $t-1$ we have a set of particles $\{X_{t-1}^{(i)}\}_{i=1}^{N}$, and weights $\{w_{t-1}^{(i)}\}_{i=1}^N$, which produce a discrete approximation to $p(x_{t-1}|y_{1:t-1},\theta)$, we can then create a Monte Carlo approximation for \eqref{eq:6} as
\begin{equation}
  \label{eq:7}
  p(x_t|y_{1:t},\theta) \approx c  g_\theta(y_t|x_t) \sum_{i=1}^{N} w_{t-1}^{(i)} f_\theta(x_t|x_{t-1}^{(i)}),
\end{equation}
where $c$ is a normalising constant. Particle approximations as given above can be updated recursively by propagating and updating the particle set using importance sampling techniques. There is now an extensive literature on particle filtering algorithms, see for example, \cite{Doucet2000} and \cite{Cappe2007}.

In this paper the particle approximations of the latent process are created with the auxiliary particle filter of \cite{Pitt1999a}.  This filter has a general form, and simpler filters can be derived as special 
cases \citep{Fearnhead2007}. The idea is to approximate
%\begin{equation*}
  $c w_{t-1}^{(i)}g_\theta(y_t|x_t)f_\theta(x_t|x_{t-1}^{(i)})$
%\end{equation*}
with
%\begin{equation*}
$ \xi_t^{(i)}q(x_t|x_{t-1}^{(i)},y_t,\theta)$,
%\end{equation*}
for a set of probabilities $\xi_t^{(i)}$ and proposal densities $q(x_t|x_{t-1}^{(i)},y_t,\theta)$.
We simulate particles at time $t$ by first choosing a particle at time $t-1$, with particle $x_{t-1}^{(i)}$ being chosen with probability $\xi_t^{(i)}$. 
We then propagate this to time $t$ by sampling our particle at time $t$, $x_t$, from $q(x_t|x_{t-1}^{(i)},y_t,\theta)$. The importance sampling weight assigned to our new particle $x^{(i)}_t$ is then
$w_{t-1}^{(i)}g_\theta(y_t|x_t)f_\theta(x_t|x_{t-1}^{(i)})/\xi_t^{(i)}q(x_t|x_{t-1}^{(i)},y_t,\theta)$.
%\[
%\frac{w_{t-1}^{(i)}g_\theta(y_t|x_t)f_\theta(x_t|x_{t-1}^{(i)})}{\xi_t^{(i)}q(x_t|x_{t-1}^{(i)},y_t,\theta)}.
%\]
%This can be shown to be a valid importance sampling weight by viewing both the proposal and target as densities on the joint distribution of the state at time $t$ and the particle at time $t-1$. 
Details are summarised in Algorithm \ref{alg1}.

\begin{algorithm}
\caption{Auxiliary Particle Filter}          % give the algorithm a caption
\label{alg1}                           % and a label for \ref{} commands later in the document
 \textit{Step 1:} iteration $t=1$. \\
 \quad Sample $\{x_1^{(i)}\}$ from the prior $p(x_1|\theta)$, set and normalise weights $w_1^{(i)} = g_\theta(y_1|x_1^{(i)})$. \\
 \textit{Step 2:} iteration $t=2,\ldots,T$.\\
 Assume a set of particles $\{x_{t-1}^{(i)}\}_{i=1}^N$ and associated weights $\{w_{t-1}^{(i)}\}_{i=1}^N$ that approximate $p(x_{t-1}|y_{1:t-1},\theta)$ and user-defined set of proposal weights $\{\xi_{t}^{(i)}\}_{i=1}^N$ and family of proposal densities $q(\cdot|x_{t-1},y_t,\theta)$. 

 \quad (a) Sample indices $\{k_1,k_2,\ldots,k_N\}$ from $\{1,\ldots,N\}$ with probabilities $\xi_{t}^{(i)}$. 

 \quad (b) Propagate particles $x_t^{(i)} \sim q(\cdot|x_{t-1}^{(k_i)},y_{t},\theta)$. 

 \quad (c) Weight each particle $w_t^{(i)} \propto \frac{w_{t-1}^{(k_i)} g_\theta(y_t|x_t^{(i)})f_\theta(x_t^{(i)}|x_{t-1}^{(k_i)})}{\xi_t^{(k_i)}q(x_t^{(i)}|x_{t-1}^{(k_i)},y_t,\theta)}$ and normalise the weights.% such that $\sum_{i=1}^N w_t^{(i)}=1$.
\end{algorithm}

\section{Parameter estimation for state space models}
\label{sec:param-estim-state}

\subsection{Maximum likelihood estimation}
\label{sec:maxim-likel-estim-1}

The maximum likelihood approach to parameter estimation is based on solving
\begin{equation*}
  \label{eq:14}
  \hat{\theta} = \arg\max_{\theta \in \Theta} \ \log p(y_{1:T}|\theta)= \arg\max_{\theta \in \Theta}\sum_{t=1}^{T} \log p(y_t|y_{1:t-1},\theta),
\end{equation*}
where,
\begin{equation*}
  p(y_t|y_{1:t-1},\theta) = \int \left(g_\theta(y_t|x_t) \int f_\theta(x_t|x_{t-1}) p(x_{t-1}|y_{1:t-1},\theta) dx_{t-1} \right) dx_t.
\end{equation*}
Aside from a few simple cases, it is not possible to calculate the log-likelihood in closed form. Pointwise estimates of the log-likelihood can be obtained using SMC approximations \citep{hurzler2001} for a fixed value $\theta$. If the parameter space $\Theta$ is discrete and low dimensional, then it is relatively straightforward to find the $\theta$ which maximises $\log p(y_{1:T}|\theta)$. For problems where the parameter space is continuous, finding the maximum likelihood estimate (MLE) can be more difficult. One option is to evaluate the likelihood over a grid of $\theta$ values, but this is computationally inefficient when the model dimension is large.

The gradient based method for parameter estimation, also known as the steepest ascent algorithm, maximises the log-likelihood function by evaluating the score vector (gradient of the log-likelihood) at the current parameters and then moving them in the direction of the gradient. For a given batch of data $y_{1:T}$, the unknown parameter $\theta$ can be estimated by choosing an initial estimate $\theta_0$, and then recursively solving
\begin{equation}
\label{eq:35}
\theta_{k} = \theta_{k-1} + \gamma_{k} \nabla \log p(y_{1:T}|\theta)|_{\theta=\theta_{k-1}}
\end{equation}
until convergence. 
Here $\gamma_k$ is a sequence of decreasing step sizes which satisfies the conditions $\sum_k \gamma_k = \infty$ and $\sum_k \gamma_k^2 < \infty$. One common choice is $\gamma_k=k^{-\alpha}, \mbox{where} \ 0.5<\alpha<1$. %(e.g. $\gamma_k=k^{-2/3}$). 
The conditions on $\gamma_k$ are necessary to ensure convergence to a value $\hat\theta$ for which $\nabla \log p(y_{1:T}|\hat\theta)=0$. A key ingredient to good statistical properties of the resulting estimator of $\theta$, such as consistency \citep{Crowder1986}, is that if the data are generated from $ p(y_{1:T}|\theta^*)$, then
\[
\Expect{\nabla \log p(Y_{1:T}|\theta^*)}=\int p(y_{1:T}|\theta^*) \nabla \log p(y_{1:T}|\theta^*)\mbox{d}y_{1:T}=0.
\]
That is, the expected value of $ \nabla \log p(y_{1:T}|\theta)$, with expectation taken with respect to the data, is 0 when $\theta$ is the true parameter value.

The rate of convergence of (\ref{eq:35}) can be improved if we are able to calculate the observed information matrix, which provides a measure of the curvature of the log-likelihood. When this is possible the Newton-Raphson method can be used and the step size parameter $\gamma_k$ is replaced with $-\gamma_k\{\nabla^2 \log p(y_{1:T}|\theta)\}^{-1}$.

\subsection{Estimation of the score and observed information matrix}
\label{sec:estim-score-vect}

For nonlinear and non-Gaussian state space models it is impossible to derive the score and observed information exactly. In such cases, SMC can be used to produce particle approximations in their place \citep{Poyiadjis2011}. If we assume that it is possible to obtain a particle approximation of the latent process $p(x_{1:T}|y_{1:T},\theta)$, then this approximation can be used to estimate the score vector $\nabla \log p(y_{1:T}|\theta)$ using Fisher's identity \citep{cappe2005inference}
\begin{equation}
\label{eq:23}
  \nabla \log p(y_{1:T}|\theta) = \int \nabla \log p(x_{1:T},y_{1:T}|\theta)p(x_{1:T}|y_{1:T},\theta)dx_{1:T}.
\end{equation}

A similar identity for the observed information matrix is given by \cite{Louis1982}
\begin{equation}
\label{eq:24}
  -\nabla^2 \log p(y_{1:T}|\theta) = \nabla \log p(y_{1:T}|\theta) \nabla \log p(y_{1:T}|\theta)^{\top} - \frac{\nabla^2p(y_{1:T}|\theta)}{p(y_{1:T}|\theta)},
\end{equation}
where,
\begin{align}
\label{eq:25}
 \frac{\nabla^2p(y_{1:T}|\theta)}{p(y_{1:T}|\theta)} = & \int \nabla \log p(x_{1:T},y_{1:T}|\theta) \nabla \log p(x_{1:T},y_{1:T}|\theta)^{\top} p(x_{1:T}|y_{1:T},\theta)dx_{1:T} \\
+ &  \int \nabla^2 \log p(x_{1:T},y_{1:T}|\theta)p(x_{1:T}|y_{1:T},\theta)dx_{1:T}. \nonumber
\end{align}
See \cite{cappe2005inference} for further details of both identities.

If we assume that the conditional densities \eqref{eq:2} and \eqref{eq:3} are twice continuously differentiable, then from the joint density \eqref{eq:4} we get
\begin{equation}
 \label{eq:36}
\nabla\log p(x_{1:T},y_{1:T}|\theta)=   \sum_{t=1}^{T} \left\{\nabla \log g_\theta(y_t|x_t) + \nabla \log f_\theta(x_t|x_{t-1})\right\},
\end{equation}
where we introduce the notation $f_\theta(x_1|x_{0})= \mu_\theta(x_1)$ to give a simpler form and similarly for the second derivative we have 
\begin{equation}
\label{eq:50}
  \nabla^2 \log p(x_{1:T},y_{1:T}|\theta)=  \sum_{t=1}^{T}\left\{ \nabla^2 \log g_\theta(y_t|x_t) + \nabla^2 \log f_\theta(x_t|x_{t-1})\right\}.
\end{equation}
In the next section we shall introduce a sequential Monte Carlo algorithm which creates approximations of these terms.

\section{Particle approximations of the score vector and observed information matrix}
\label{sec:kern-dens-estim}

\subsection{Kernel density methods to overcome particle degeneracy}

In this section we focus on applying our method to the score vector $\nabla\log p(y_{1:t}|\theta)$ and note that extending these results to the observed information matrix is straightforward and not given explicitly (see Algorithm \ref{alg6} for implementation details). Using a particle filter (Alg. \ref{alg1}) we can sample $x_t^{(i)}$ and let $x_{1:t}^{(i)}$ denote the path associated with that particle. At time $t$ particle $i$ stores value $\alpha_t^{(i)}=\nabla\log p(x_{1:t}^{(i)},y_{1:t}|\theta)$, which depends on the history of the particle, $x_{1:t}^{(i)}$. The estimate for $\alpha_t$ is then updated recursively, where at iteration $t$ we have particles $x_t^{(i)}$ with associated weights $w_t^{(i)}$. If we assume that particle $i$ is descended from particle $k_i$ at time $t-1$, then \eqref{eq:36} can be given as
\begin{equation} \label{eq:a1}
\alpha_t^{(i)}=\alpha_{t-1}^{(k_i)} + \nabla \log g_\theta(y_t|x_t^{(i)}) + \nabla \log f_\theta(x_t^{(i)}|x_{t-1}^{(k_i)}).
\end{equation}
The score vector $S_t=\nabla\log p(y_{1:t}|\theta)$ at time $t$ is then approximated as 
\[
S_t = \sum_{i=1}^N w_t^{(i)}\alpha_t^{(i)}.
\]
Estimation of the score vector in this fashion does not require that we store the entire path of the latent process $\{X_{1:T}^{(i)}\}_{i=1}^N$. However, the $\alpha_t^{(i)}s$  that are stored for each particle depend on the complete path-history of the associated particle. Particle approximations of this form are known to be poor due to inherent particle degeneracy over time \citep{Andrieu2005}. \cite{Poyiadjis2011} prove that the asymptotic variance of the estimate of the score vector increases at least quadratically with time. This can be attributed to the standard problem of particle degeneracy in particle filters when approximating the conditional distribution of the complete path of the latent state $p(x_{1:t}|y_{1:t})$. One approach to reduce this degeneracy is to use kernel density methods, such as the \cite{West} algorithm, which we apply here to the $\alpha_t^{(i)}$s.

The idea of \cite{West} is to combine shrinkage of the $\alpha_t^{(i)}$s towards their mean, together with adding noise. The latter is necessary for overcoming particle degeneracy, but the former is required to avoid the increasing variance of the $\alpha_t^{(i)}$s. Implementing this strategy we start by replacing $\alpha_{t-1}^{(k_i)}$ with a draw from a Gaussian kernel, where $k_i$ is drawn from a discrete distribution with probabilities $\xi_{t}^{(i)}$, and where the mean and variance of $\alpha_{t-1}^{(k_i)}$ are
\[
{S}_{t-1}=\sum_{i=1}^N w_{t-1}^{(i)}\alpha_{t-1}^{(i)} \quad\mbox{and}\quad  \Sigma_{t-1}^{\alpha}=\sum_{i=1}^N w_{t-1}^{(i)} (\alpha_{t-1}^{(i)}-S_{t-1})^{\top}(\alpha_{t-1}^{(i)}-S_{t-1}).
\]

If we let $0<\lambda<1$ be a shrinkage parameter, which is a fixed constant, and choose a density bandwidth $h>0$, we can replace $\alpha_{t-1}^{(k_i)}$ in (\ref{eq:a1}) with
\begin{equation} 
\label{eq:KDE}
\lambda \alpha_{t-1}^{(k_i)}  + (1-\lambda) S_{t-1}+\epsilon_t^{(i)},
\end{equation}
where $\epsilon_t^{(i)}$ is a realisation of a Gaussian distribution $\mathcal{N}(0,h^2\Sigma_{t-1}^{\alpha})$. By choosing $\lambda$ and $h$ such that $\lambda^2+h^2=1$ \citep{West}, it is then straightforward to show that this kernel density approximation preserves the mean and variance of the $\alpha_t^{(i)}$s.

\subsection{Rao-Blackwellisation}
\label{sec:rao-blackwellisation}

The stored $\alpha_t^{(i)}$ values do not have any effect on the dynamics of the state. Furthermore, we have a stochastic update for these terms which, when we use the kernel density approach, results in a linear-Gaussian update. This means that we can use the idea of  Rao-Blackwellisation  \cite[]{Doucet2000} to reduce the variance in our estimates of the score vector and observed information matrix. In practice this means replacing the $\alpha_t^{(i)}$ values by an appropriate distribution which is sequentially updated. Therefore we do not need to add noise to the approximation at each time step as we do with the standard kernel density approach. Instead we can recursively update the mean and variance of the distribution representing $\alpha_t^{(i)}$ and estimate the score vector $S_t$.

For $t\geq 2$, assume that at time $t-1$ each $\alpha_{t-1}^{(j)}$ is represented by a Gaussian distribution, 
\[
\alpha_{t-1}^{(j)} \sim \mathcal{N}(m_{t-1}^{(j)},h^2 V_{t-1}).
\]
Then from (\ref{eq:a1}) and (\ref{eq:KDE}) we have that
\begin{equation}
  \label{eq:9}
\alpha_{t}^{(i)} \sim \mathcal{N}(m_{t}^{(i)},h^2 V_{t}),  
\end{equation}
where,
\[
m_t^{(i)}=\lambda m_{t-1}^{(k_i)}+(1-\lambda)S_{t-1} + \nabla \log g_\theta(y_t|x_t^{(i)}) + \nabla \log f_\theta(x_t^{(i)}|x_{t-1}^{(k_i)}),
\]
and 
\[
V_t= V_{t-1}+\Sigma_{t-1}^{\alpha}=V_{t-1}+\sum_{i=1}^N w_{t-1}^{(i)} (m_{t-1}^{(i)}-S_{t-1})^{\top}(m_{t-1}^{(i)}-S_{t-1}).
\]
%\[
%V_t= V_{t-1}+\sum_{i=1}^N w_{t}^{(i)} (m_{t}^{(i)}-S_{t})^{\top}(m_{t}^{(i)}-S_{t}).
%\]

The estimated score vector at each iteration is a weighted average of the $\alpha_t^{(i)}$s, so we can  estimate  the score by
\begin{equation}
  \label{eq:10}
 S_t   =  \sum_{i=1}^N w_t^{(i)} m_{t}^{(i)}.  
\end{equation}
If we only want to estimate the score vector, then this shows that we only need to calculate the expected value of the $\alpha_t^{(i)}$s. However, if we wish to calculate the observed information matrix $I_t$, then from \eqref{eq:25}, a standard particle approximation would give
\[
I_t  = S_t S_t^{\top}-\sum_{i=1}^N w_t^{(i)}\left\{\alpha_t^{(i)}\alpha_t^{{(i)}^{\top}}+\beta_t^{(i)} \right\},  \\
\]
where we define $\beta_t^{(i)} = \nabla^2\log p(x_{1:t}^{(i)},y_{1:t}|\theta)$. Taking the same approach for $\beta_t^{(i)}$ as we did for $\alpha_t^{(i)}$, we define a Gaussian distribution for $\beta_t^{(i)}$ and update its mean and covariance in the same way as was shown above for $\alpha_t$. In practice we only need to calculate the mean, which we will denote as $n_t^{(i)}$. Using Rao-Blackwellisation, and the assumed distributions for $\alpha_t^{(i)}$ and $\beta_t^{(i)}$, gives the following estimate of the observed information matrix
\[
I_t  = S_t S_t^{\top}-\sum_{i=1}^N w_t^{(i)}\left\{m_t^{(i)}m_t^{{(i)}^{\top}}+h^2V_t+n_t^{(i)} \right\}.  \\
\]

Note the inclusion of $h^2V_t$ in this estimate. This term is important as it corrects for the fact that shrinking the values of $\alpha_t$ towards $S_t$ at each iteration will reduce the variability in these values. Without this correction the observed information would be overestimated. Details of this approach are summarised in Algorithm \ref{alg6}.
%Kernel Density Vector Approximation of the
\begin{algorithm}
\caption{Rao-Blackwellised Score and Observed Information Matrix}          % give the algorithm a caption
\label{alg6}                           % and a label for \ref{} commands later in the document
 \textit{Initialise:} set $m_0^{(i)}=0$ and $n_0^{(i)}=0$ for $i=1\ldots,N$, $S_0=0$ and $B_0=0$. \\
At iteration $t = 1,\ldots, T$, 

 \quad (a)  Apply Algorithm \ref{alg1} to obtain $\{x_{t}^{(i)}\}_{i=1}^N$, $\{k_i\}_{i=1}^N$ and $\{w_t^{(i)}\}_{i=1}^N$ 

 \quad (b) Update the mean of the approximations for $\alpha_t$ and $\beta_t$ 
\begin{eqnarray*}
m_t^{(i)}= & \lambda m_{t-1}^{(k_i)} + (1-\lambda) S_{t-1}  +  \nabla \log g_\theta(y_t|x_t^{(i)}) + \nabla \log f_\theta(x_t^{(i)}|x_{t-1}^{(k_i)}) \\
n_t^{(i)} = & \lambda n_{t-1}^{(k_i)} + (1-\lambda) B_{t-1} + \nabla^2 \log g_\theta(y_t|x_t^{(i)}) + \nabla^2 \log f_\theta(x_t^{(i)}|x_{t-1}^{(k_i)}) 
\end{eqnarray*}
 \quad (b) Update the score vector and observed information matrix 
\begin{eqnarray*}
S_t = \sum_{i=1}^N w_t^{(i)}m_{t}^{(i)} \quad \mbox{and} \quad
I_t=S_t S_t^{\top} -\sum_{i=1}^N w_t^{(i)}(m_t^{(i)} m_t^{{(i)}^{\top}} + n_t^{(i)}) - h^2V_t
\end{eqnarray*}
\mbox{where} $V_t= V_{t-1}+\sum_{i=1}^N w_{t-1}^{(i)} (m_{t-1}^{(i)}-S_{t-1})^{\top}(m_{t-1}^{(i)}-S_{t-1})$ and $B_t=\sum_{i=1}^N w_t^{(i)}n_t^{(i)}.$ 
\end{algorithm}

Our new $\mathcal{O}(N)$ algorithm can be viewed as a generalisation of the \cite{Poyiadjis2011} algorithm. Setting $\lambda=1$ in Algorithm \ref{alg6} gives the Poyiadjis algorithm. 
However, this algorithm, as illustrated in Section \ref{sec:comp-appr} and proved by \cite{Poyiadjis2011}, has a quadratically increasing variance in $t$. 
As a result, \cite{Poyiadjis2011} introduce an alternative algorithm whose computational cost is quadratic in the number of particles, but which has better Monte Carlo properties. 
\cite{DelMoral2010} and \cite{douc2011} show that this alternative approach, under standard mixing assumptions, produces estimates of the score with an asymptotic variance that increases only linearly with time. %Details of this algorithm are omitted for brevity, for further details see \cite{Poyiadjis2011}.

\section{Theoretical justification}
\label{sec:theor-just}
\subsection{Monte Carlo accuracy}

We have motivated the use of both the kernel density approximation and Rao-Blackwellisation as a means to reduce the impact of particle degeneracy on the $\mathcal{O}(N)$ algorithm for 
estimating the score vector and observed information matrix. However, what can we say about the resulting algorithm? 
%Here we consider both the Monte Carlo accuracy of the resulting algorithm, and the effect of the approximation error within the kernel density approximation in terms of inferences for the parameters.

It is possible to implement Algorithm \ref{alg6} so as to store the whole history of the state $x_{1:t}$, rather than just the current value, $x_t$. This just involves extra storage, with our particles being $x_{1:t}^{(i)}=(x_t^{(i)},x_{1:t-1}^{(k_i)})$. Whilst unnecessary in practice, thinking about such an algorithm helps with understanding the algorithms properties.

One can fix $\theta$, the parameter value used when running the particle filter algorithm, and the data $y_{1:t}$. For convenience we drop the dependence on $\theta$ from notation in the following. The $m_t^{(i)}$ values calculated by the algorithm are just functions of the history of the state and the past estimated score values. We can define a set of functions $\phi_s(x_{1:t})$, 
\[
\phi_s(x_{1:t})= \nabla \log g_\theta(y_s|x_s) + \nabla \log f_\theta(x_s|x_{s-1}),
\]
where $t\geq s>0$ and functions, $m_s(x_{1:t})$, which depend on $m_{s-1}(x_{1:t})$ and the estimated score functions at previous time-steps, $S_{0:s-1}$, through 
\begin{equation}\label{eq:1a}
m_s(x_{1:t})=\lambda m_{s-1}(x_{1:t}) +(1-\lambda) S_{s-1} + \phi_s(x_{1:t}),
\end{equation}
with $m_0(x_{1:t})=0$. We then have that in Algorithm \ref{alg6}, $m_t^{(i)}=m_t(x_{1:t}^{(i)})$, is the value of this function evaluated for the state history associated with the $i$th particle at time $t$. 

Note that it is possible to iteratively solve the recursion (\ref{eq:1a}) to get
\begin{equation}\label{eq:1b}
m_s(x_{1:t})= \sum_{u=1}^s \lambda^{s-u} \phi_u(x_{1:t})+ (1-\lambda)\sum_{u=1}^s \lambda^{s-u}S_{u-1}
\end{equation}
where $0<\lambda<1$ is the shrinkage parameter.

If we set $\lambda=1$, then Algorithm \ref{alg6} reverts to the Poyiadjis $\mathcal{O}(N)$ algorithm and \eqref{eq:1b} simplifies to a sum of additive functionals $\phi_u(x_{1:t})$. The poor Monte Carlo properties of this algorithm stem from the fact that the Monte Carlo variance of SMC estimates of $\phi_u(x_{1:t})$ increase at least linearly with $s-u$. And hence the Monte Carlo variance of the SMC estimate of  $\sum_{u=1}^s \phi_u(x_{1:t})$, increases at least quadratically with $s$. 

In terms of the Monte Carlo accuracy of Algorithm \ref{alg6}, the key is that in (\ref{eq:1b}) we exponentially down-weight the contribution of $\phi_u(x_{1:t})$ as $s-u$ increases. Under quite weak assumptions, such as the Monte Carlo variance of the estimate of  $\phi_u(x_{1:t})$ being bounded by a polynomial in $s-u$, we will have that the Monte Carlo variance of estimates of $ \sum_{u=1}^s \lambda^{s-u} \phi_u(x_{1:t})$ will now be bounded in $s$.

For $\lambda<1$, we introduce the additional second term in (\ref{eq:1b}), without which there would be a substantial bias in the score estimate that would grow with $t$. Estimating this term is less problematic as the Monte Carlo variance of each $S_{u-1}$ will depend only on $u$, and will not increase as $s$ increases. Empirically, the resulting Monte Carlo variance of our estimates of the score increase only linearly with $s$ for a wide-range of models.

\subsection{Effect on parameter inference}

Now consider the  value of  $S_t$ in the limit as the number of particles goes to infinity, $N\rightarrow \infty$. We assume that standard conditions on the particle filter for the law of large numbers \citep{Chopin2004} hold. Then we have that
\[
S_t \rightarrow \Expects{\theta}{m_t(X_{1:t})|y_{1:t}}=\int m_t(x_{1:t})p(x_{1:t}|y_{1:t},\theta)\mbox{d}x_{1:t}.
\]

For $t=1,\ldots,T$, where we fix the data $y_{1:T}$, define $\bar{S}_t=\Expects{\theta}{m_t(X_{1:t})|y_{1:t}}$ to be the large $N$ limit of the estimate of the score at time $t$.  The following lemma expresses $\bar{S}_t$ in terms of expectations of the $\phi_s(\cdot)$ functions. Proofs from this section can be found in the supplementary material.
\begin{lemma}
Fix $y_{1:T}$. Then $\bar{S}_1= \Expects{\theta}{\phi_1(X_{1:t})|y_{1}}$ and
for $2\leq t \leq T$
\[
\bar{S}_t= \sum_{u=1}^t \lambda^{t-u} \Expects{\theta}{\phi_u(X_{1:t})|y_{1:t}} +
(1-\lambda)\sum_{u=1}^{t-1} \sum_{s=u}^{t-1} 
%\left(\begin{array}{c}t-u \\ t-s \end{array}\right)
\lambda^{s-u}
\Expects{\theta}{\phi_u(X_{1:t})|y_{1:s}},
\]
where the expectations are taken with respect to the conditional distribution of $X_{1:t}$ given $y_{1:u}$:
\[
\Expects{\theta}{\phi_s(X_{1:t})|y_{1:u}} = \int \phi_s(x_{1:t})p(x_{1:t}|y_{1:u},\theta)\mbox{d}x_{1:t}.
\]
\end{lemma}

We now consider taking expectation of $\bar{S}_T$ with respect to the data. We write $\bar{S}_T(y_{1:T};\theta)$ to denote the dependence on the data $y_{1:T}$ and the choice of parameter $\theta$ when implementing the particle filter algorithm. A direct consequence of Lemma 1 is the following theorem.
\begin{theorem}
\label{thr:1}
Let $\theta^*$ be the true parameter value, and $T$ a positive integer. Assume regularity conditions exist so that for all $t\leq T$,
\begin{equation} \label{eq:t1}
\Expects{\theta^*}{\nabla \log p(X_{1:t},Y_{1:t}|\theta^*)}=0,
\end{equation}
where expectation is taken with respect to $p(X_{1:T},Y_{1:T}|\theta^*)$. Then
\[
\Expects{\theta^*}{\bar{S}_T(Y_{1:T};\theta^*)}=0,
\]
where expectation is taken with respect to $p(Y_{1:T}|\theta^*)$. 
\end{theorem}

The theorem shows that for any $0<\lambda<1$, the expectation of $\bar{S}_T(y_{1:T};\theta^*)$ at the true parameter $\theta^*$ is zero, and hence $\bar{S}_T(y_{1:T};\theta)=0$ are a set of unbiased estimating equations for $\theta$. Using our estimates of the score function within the steepest gradient ascent algorithm is thus using Monte Carlo estimates to approximately solve this set of unbiased estimating equations. 

The accuracy of the final estimate of $\theta$ will depend both on the amount of Monte Carlo error, and also the accuracy of the estimator based on solving the underlying estimating equation. Note that the statistical efficiency of the estimator obtained by solving $\bar{S}_T(y_{1:T};\theta)=0$ may be different, and lower, than that of solving $\nabla \log p(y_{1:T}|\theta)=0$. However in practice we would expect this to be more than compensated by the reduction in Monte Carlo error we get. We investigate this empirically in the following sections.

\section{Comparison of approaches}
\label{sec:comp-appr}

In this section we shall evaluate our algorithm and compare existing approaches for estimating the score vector. Most importantly, we will investigate how the performance of our method depends on the choice of shrinkage parameter, $\lambda$. For comparison, we consider a linear-Gaussian state space model, where it is possible to analytically calculate the score vector and observed information matrix using a Kalman filter \citep{Kalman1960}. 

Consider a first order autoregressive model AR(1) observed with Gaussian noise:
\begin{equation}
  \label{eq:15}
  Y_t|X_t=  x_t \sim  \mathcal{N}(x_t,\tau^2), \quad X_t | X_{t-1}=  x_{t-1} \sim \mathcal{N}(\phi x_{t-1},\sigma^2),\quad X_1 \sim \mathcal{N}\left(0,\frac{\sigma^2}{1-\phi^2}\right),
\end{equation}
where we can derive the optimal proposal distribution for the particle filter
\begin{equation*}
  \label{eq:16}
q(x_t|x_{t-1}^{(i)},y_t) = \mathcal{N}\left(x_t \bigg| \frac{\phi x_{t-1}^{(i)} \tau^2 + y_t \sigma^2}{\sigma^2+\tau^2},\frac{\sigma^2 \tau^2}{\sigma^2+\tau^2}\right), ~~
\xi_t^{(i)} \propto w_{t-1}^{(i)} \mathcal{N}(y_t|\phi x_{t-1}^{(i)},\sigma^2+\tau^2). 
\end{equation*}

We shall compare our algorithm (Alg. \ref{alg6}) against the $\mathcal{O}(N)$ and $\mathcal{O}(N^2)$ algorithms of \cite{Poyiadjis2011}, and also the fixed-lag smoother of \cite{kitagawa2001}.

The fixed-lag smoother is based on approximating $p(x_{1:t}|y_{1:T},\theta)$ with $p(x_{1:t}|y_{1: \min\{t+L,T\}},\theta)$, where $L$ is some pre-specified lag. The posterior, $p(x_{1:t}|y_{1: \min\{t+L,T\}},\theta)$, can then be estimated using an $\mathcal{O}(N)$ algorithm. This method reduces the Monte Carlo variance at the cost of introducing a bias. Theoretical results given by \cite{Olsson2008a} show that as $T$ increases the optimal choice of $L$, in terms of a bias-variance trade-off, is $O(\log(T))$.

We perform a comparison on a data set of length $T=20,000$ simulated from the autoregressive model \eqref{eq:15} with parameters $\theta^*=(\phi,\sigma,\tau)^\top=(0.8,0.5,1)^\top$. Our method and the Poyiadjis $\mathcal{O}(N)$ have the same computational cost and are implemented with $N=50,000$. The Poyiadjis $\mathcal{O}(N^2)$ algorithm, which has a quadratic computational cost, is implemented  with $N=500$. The comparisons were run on a Dell Latitude laptop with a 1.6GHz processor, where each iteration of the $\mathcal{O}(N)$ algorithms takes approximately 1 minute for $N=50,000$. The $\mathcal{O}(N^2)$ takes 5.1 minutes for $N = 500$. This corresponds to a CPU cost that is approximately 5 times greater than the $\mathcal{O}(N)$ methods. 

\begin{figure}
  \centering
  \subfigure{\includegraphics[width=0.45\textwidth]{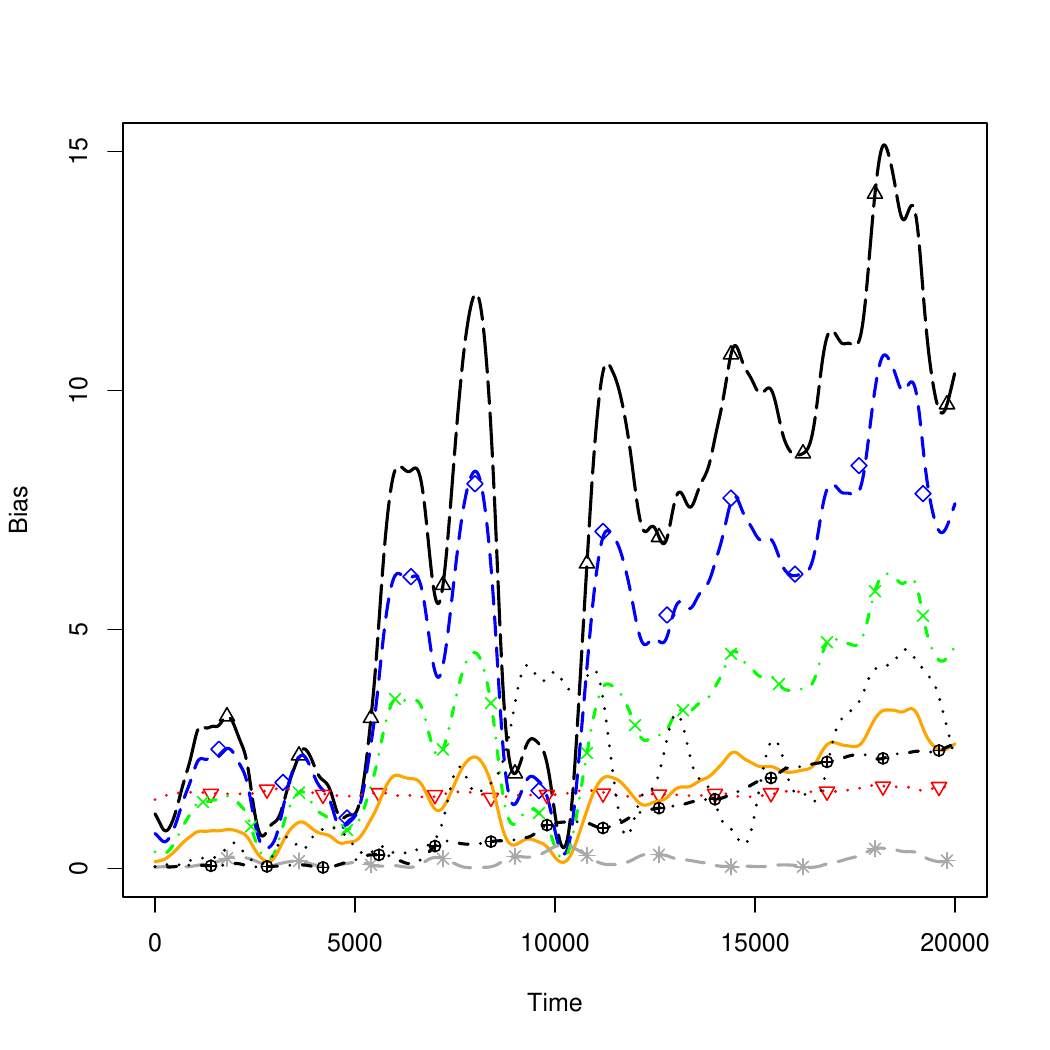}}
   \subfigure{\includegraphics[width=0.45\textwidth]{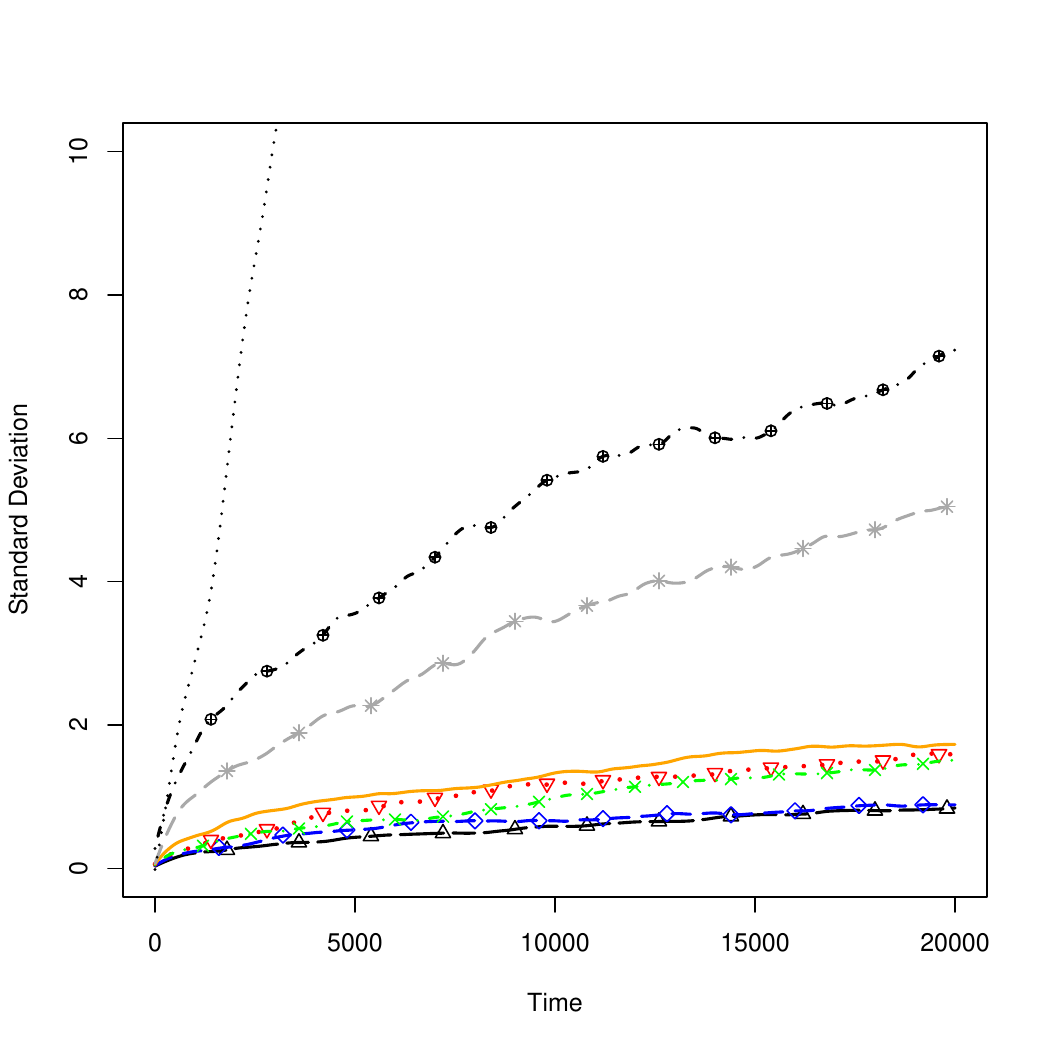}}
  \subfigure{\includegraphics[width=0.45\textwidth]{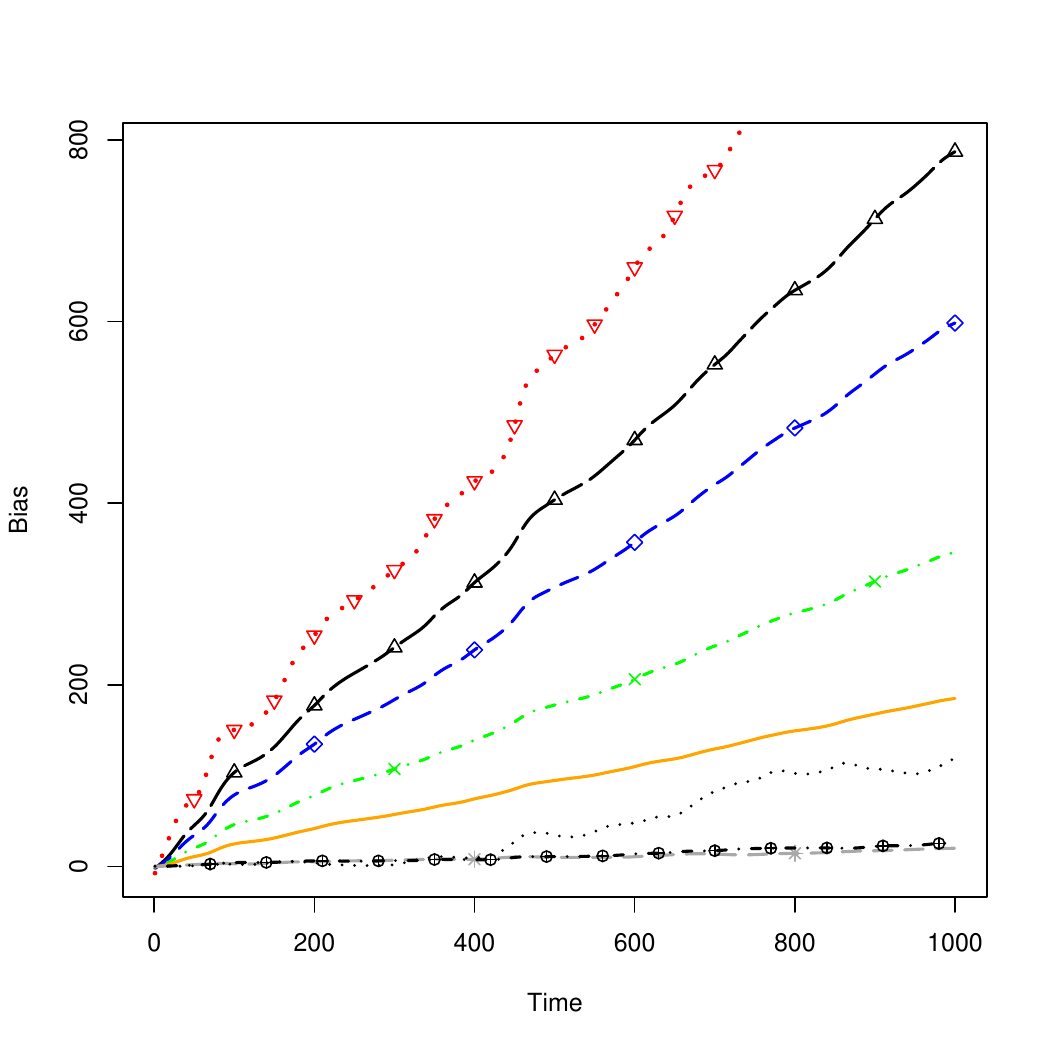}}
   \subfigure{\includegraphics[width=0.45\textwidth]{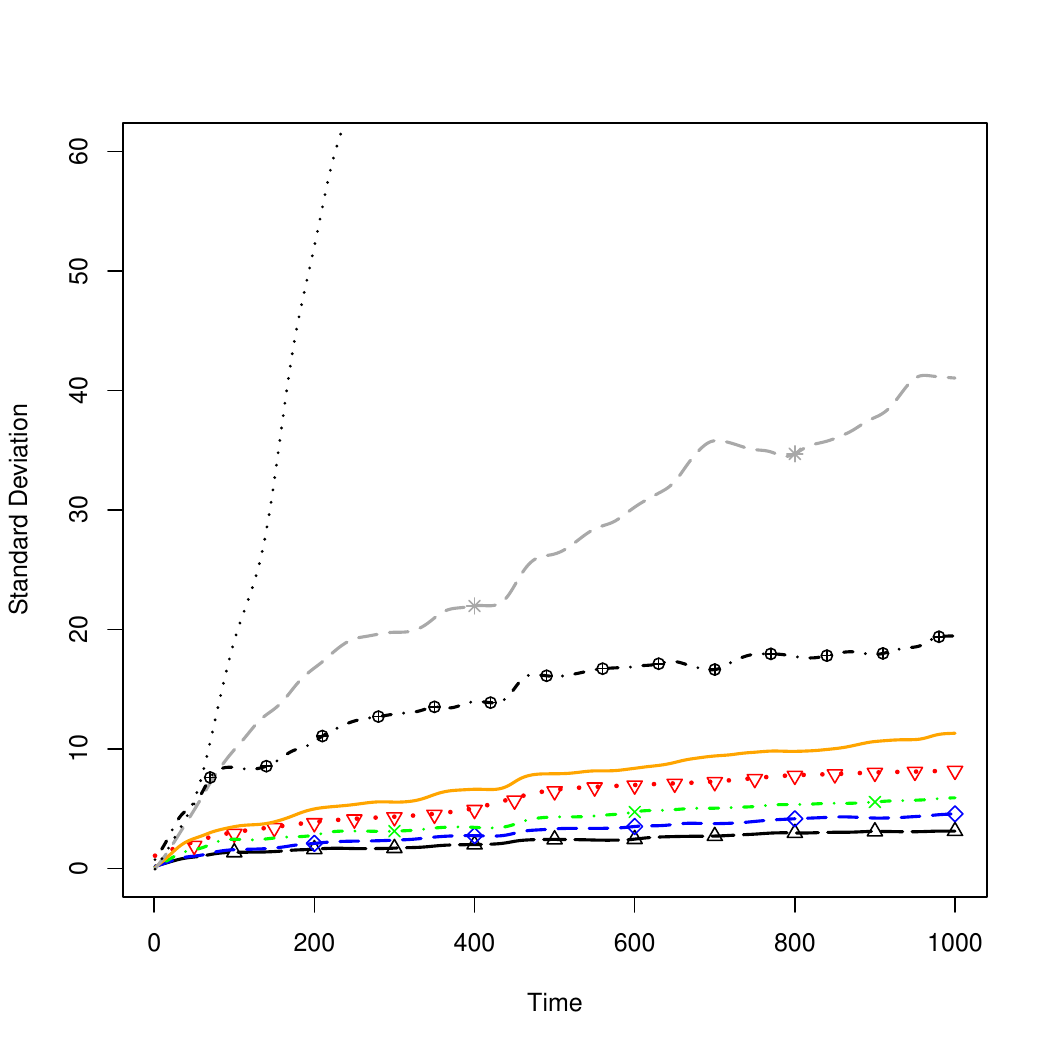}}
  \caption{Absolute bias (left column) and standard deviation (right column) of score estimates for $\tau$ (top row) and observed information matrix for the $\phi$ component (bottom row) from the autoregressive model using our $\mathcal{O}(N)$ algorithm with $\lambda = 0.99$ (\color{gray}$\protect\rule[0.5ex]{0.1cm}{0.7pt}$ $\ast$ $\protect\rule[0.5ex]{0.4cm}{0.7pt}$ $\protect\rule[0.5ex]{0.1cm}{0.7pt}$ $\ast$ $\protect\rule[0.5ex]{0.4cm}{0.7pt}$ \color{black}), $\lambda = 0.95$ (\color{orange}\protect\rule[0.5ex]{0.8cm}{0.7pt} \color{black}), $\lambda = 0.9$ (\color{green} $\boldsymbol{-} \cdot \times \boldsymbol{-} \cdot \times \boldsymbol{-}$ \color{black}), $\lambda = 0.8$ (\color{blue} $\protect\rule[0.5ex]{0.2cm}{0.7pt}$ $\Diamond$ $\protect\rule[0.5ex]{0.2cm}{0.7pt}$ $\protect\rule[0.5ex]{0.2cm}{0.7pt}$ $\Diamond$ $\protect\rule[0.5ex]{0.2cm}{0.7pt}$ \color{black}), $\lambda = 0.7$ ($\protect\rule[0.5ex]{0.1cm}{0.7pt}$ $\vartriangle$ $\protect\rule[0.5ex]{0.4cm}{0.7pt}$ $\protect\rule[0.5ex]{0.1cm}{0.7pt}$ $\vartriangle$ $\protect\rule[0.5ex]{0.4cm}{0.7pt}$), Fixed-lag smoother $L = 10$ (\color{red} $\cdot\triangledown\cdot\cdot\triangledown\cdot\cdot\triangledown\cdot$ \color{black}), and the Poyiadjis $\mathcal{O}(N)$ 
algorithm ($\boldsymbol{\cdot\cdot\cdot\cdot\cdot\cdot}$) and $\mathcal{O}(N^2)$ with $N = 500$ ($\boldsymbol{-} \cdot \otimes \boldsymbol{-} \cdot \otimes \boldsymbol{-}$).}
  \label{fig:shrink}
\end{figure}

The results given in Figure \ref{fig:shrink} show that for all but the Poyiadjis $\mathcal{O}(N)$ algorithm the standard deviation of the score estimate is increasing at a rate of $T^{-1/2}$, giving a variance that is increasing approximately linearly with time. For the Poyiadjis $\mathcal{O}(N)$, the variance is increasing quadratically (standard deviation is increasing linearly) in line with the established theoretical results. As for the $\mathcal{O}(N^2)$ algorithm, the variance increases only linearly, as expected, but at an increased computational cost compared to the $\mathcal{O}(N)$ algorithms. The variance could be further reduced by increasing the number of particles, but this will lead to a further increase in the computational cost. While the variance of the $\mathcal{O}(N^2)$ is only linearly increasing, it is worth noting that it is larger than what is given by our algorithm for all values of $\lambda$.

For estimating the score, the fixed-lag smoother performs well in terms of both bias and variance, and we note that, while not shown in Figure \ref{fig:shrink}, varying the lag about $\log(T)$ does not dramatically change the outcome, but $L=10$ seems to give the best result. However, while the fixed-lag smoother appears to work well when estimating the score, it struggles to accurately estimate the observed information, with a large bias for a range of lags ($1\leq L \leq 100$). This is because the fixed-lag approach reduces the variability in the estimates of $\nabla \log p(x_{1:t},y_{1:t}|\theta)$ associated with each particle, which means that it under-estimates the first term in Louis's identity \eqref{eq:24}. Whilst our approach also reduces the variability in the estimates of $\nabla \log p(x_{1:t},y_{1:t}|\theta)$ associated with each particle, we are able to correct for this within the Rao-Blackwellisation scheme (see Section \ref{sec:rao-blackwellisation} for details). This drawback is further explored in Section \ref{sec:struct-autor-model}.

For our algorithm, we notice that the bias and variance of both the score estimate, and observed information matrix, vary according to $\lambda$. Reducing $\lambda$ has the effect of increasing the bias, but at the same time, reducing the Monte Carlo variance of the estimates. The figures show that if we wish to minimise both bias and variance, then setting $\lambda \approx 0.95$ will produce an estimate for the score and observed information which exhibit only linearly increasing variance, with minimal bias introduced as a result. In fact, the results suggest that setting $0.9 \leq \lambda \leq 0.99$ will produce the best overall results. However, ultimately interest lies in estimating the model parameters, and in Section \ref{sec:parameter-estimation} we will see that our algorithm produces reliable estimates of the model parameters for all values of $\lambda$.

\section{Parameter estimation}
\label{sec:parameter-estimation}

Our $\mathcal{O}(N)$ algorithm, as described in Section \ref{sec:kern-dens-estim}, can be used to estimate the score vector and observed information matrix. These estimates can then be used within the steepest ascent algorithm \eqref{eq:35} to obtain the MLE for $\theta$.

The steepest ascent algorithm \eqref{eq:35} performs offline maximum likelihood estimation using batches of data $y_{1:T}$, which can be useful when dealing with small data sets. Alternatively, we could implement recursive parameter estimation, where estimates of the parameters $\theta_t$ are updated as new observations are made available. Ideally this would be achieved by using the gradient of the predictive log-likelihood,
 \begin{equation}
   \label{eq:1}
\theta_{t} = \theta_{t-1} + \gamma_{t} \nabla \log p(y_t|y_{1:t-1},\theta_t),   
 \end{equation}
where,
  \[
  \nabla \log p(y_t|y_{1:t-1},\theta_t) = \nabla \log p(y_{1:t}|\theta_t) - \nabla \log p(y_{1:t-1}|\theta_{t-1}).
\]
However, getting Monte Carlo estimates of $\nabla \log p(y_t|y_{1:t-1},\theta_t)$ is difficult due to using different values of $\theta$ at each iteration of the sequential Monte Carlo algorithm. Thus, following \cite{LeGland1997} and \cite{Poyiadjis2011}, we make a further approximation, and ignore the fact that $\theta$ changes with $t$. Instead we update $\theta_t$ at each iteration using the following approximation to this gradient:
\[
\nabla \log \hat{p}(y_t|y_{1:t-1},\theta_t) = S_t-S_{t-1}.
\]

\subsection{Autoregressive model}
\label{sec:struct-autor-model}

We compare the accuracy and efficiency of estimating the parameters of the AR(1) model \eqref{eq:15} using the various algorithms given in Section \ref{sec:comp-appr} in both an offline and online setting. Starting with the batch case (offline), we simulated 1,000 observations from the model with parameters $\theta^* = (\phi,\sigma,\tau)^\top = (0.9,0.7,1)^\top$ and estimated the score vector and observed information matrix using our $\mathcal{O}(N)$ algorithm, the fixed-lag smoother, and the $\mathcal{O}(N)$ and $\mathcal{O}(N^2)$ algorithms of Poyiadjis. The estimates of the score vector and observed information matrix were used within the Newton-Raphson algorithm \eqref{eq:35} to estimate $\theta$. The starting parameters for the algorithm are $\theta_0 = (\phi,\sigma,\tau)^\top = (0.6,1,0.7)^\top$. The AR(1) model is linear-Gaussian, and therefore allows for a direct comparison against the Kalman filter, where the score and observed information matrix can be calculated analytically.
\begin{figure}[h]
  \centering
 \subfigure{\includegraphics[width=0.45\textwidth]{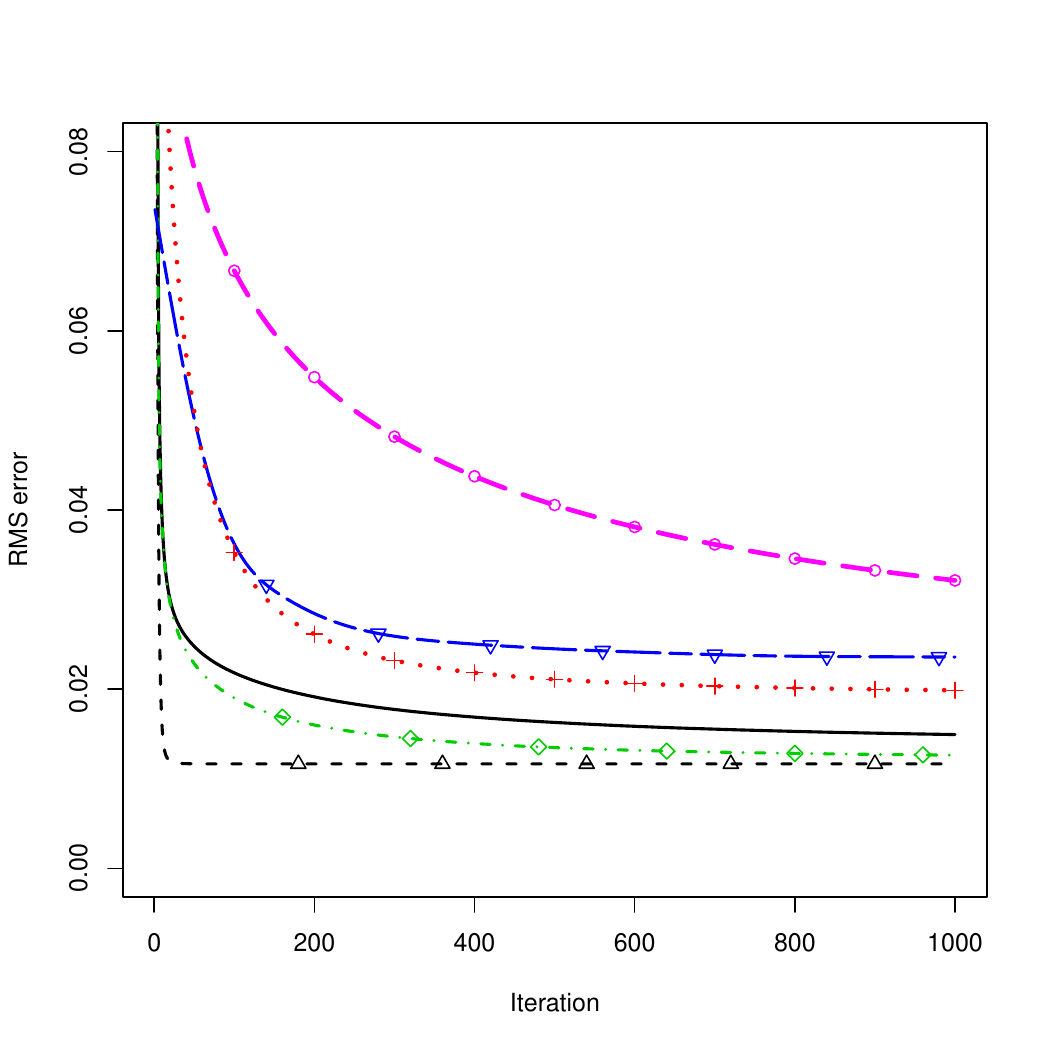}} 
   \subfigure{\includegraphics[width=0.45\textwidth]{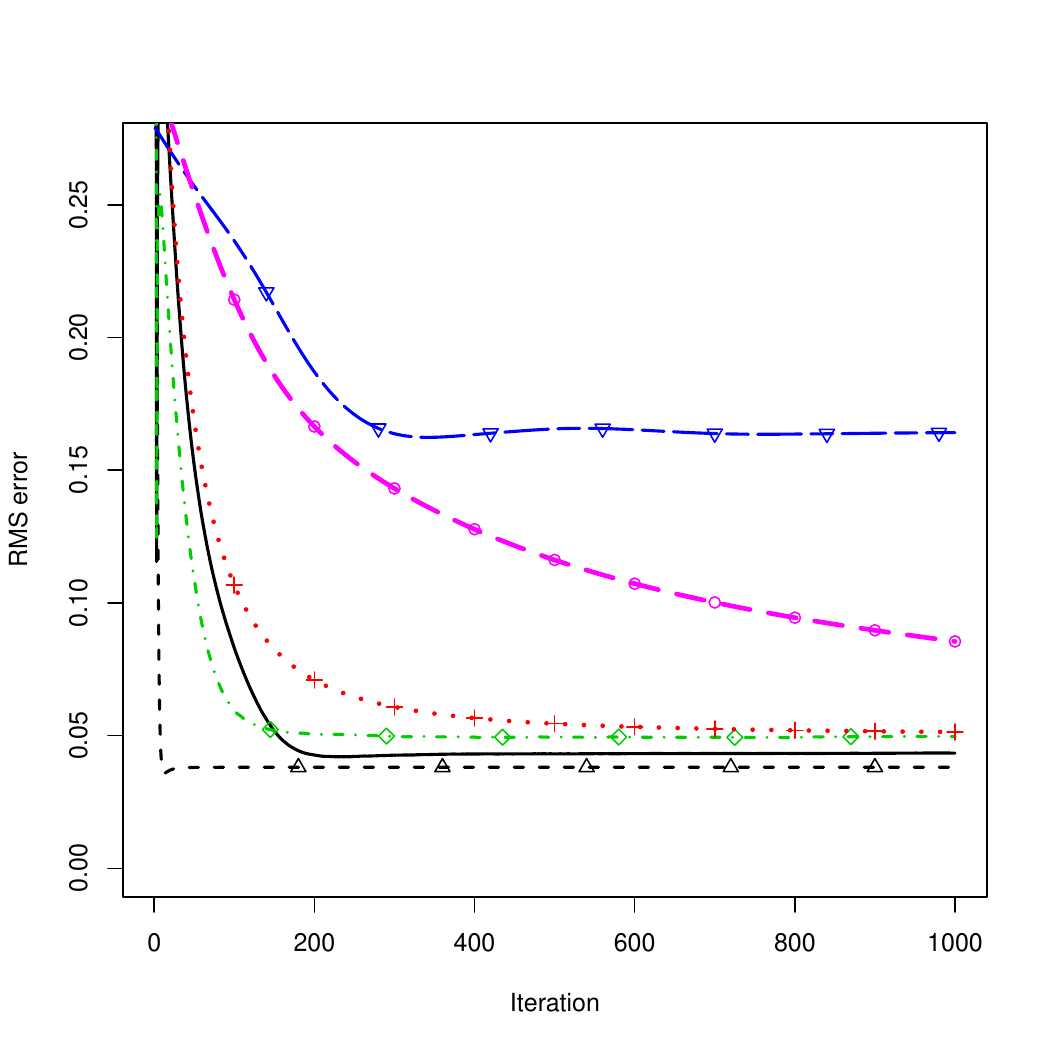}} 
   \caption{Root mean squared error of parameter estimates $\phi$ (left panel) and $\sigma$ (right panel) averaged over 20 Monte Carlo simulations from our $\mathcal{O}(N)$ algorithm with $\lambda=0.95$  (\protect\rule[0.5ex]{0.8cm}{0.6pt}), Poyiadjis $\mathcal{O}(N)$ ( \color{blue} $\protect\rule[0.5ex]{0.1cm}{0.7pt}$ $\triangledown$ $\protect\rule[0.5ex]{0.4cm}{0.7pt}$ $\protect\rule[0.5ex]{0.1cm}{0.7pt}$ $\triangledown$ $\protect\rule[0.5ex]{0.4cm}{0.7pt}$  \color{black}), Poyiadjis $\mathcal{O}(N^2)$ (\color{green} $\boldsymbol{-} \cdot \Diamond \boldsymbol{-} \cdot\boldsymbol{-} \Diamond$  \color{black}), Fixed-lag smoother ( \color{cyan}$\protect\rule[0.5ex]{0.18cm}{0.9pt}$ $\circ$ $\protect\rule[0.5ex]{0.18cm}{0.9pt}$ $\protect\rule[0.5ex]{0.18cm}{0.9pt}$ $\circ$ $\protect\rule[0.5ex]{0.18cm}{0.9pt}$ \color{black}), Fixed-lag smoother ( \color{red} $\boldsymbol{\cdot}$ $\boldsymbol{\cdot}$ $+$ $\boldsymbol{\cdot}$ $\boldsymbol{\cdot}$ $+$ $\boldsymbol{\cdot}$  \color{black})
   with score only and the Kalman filter estimate ($\protect\rule[0.5ex]{0.1cm}{0.5pt}$ 
   $\vartriangle$  $\protect\rule[0.5ex]{0.1cm}{0.5pt}$  $\protect\rule[0.5ex]{0.1cm}{0.5pt}$ $\vartriangle$ $\protect\rule[0.5ex]{0.1cm}{0.5pt}$).}
  \label{fig:ARpar_batch}
\end{figure}

Figure \ref{fig:ARpar_batch} gives the RMS error of the parameters estimated using the Newton-Raphson algorithm \eqref{eq:35} averaged over 20 Monte Carlo simulations. Our algorithm, the fixed-lag smoother and the  $\mathcal{O}(N)$ algorithm of Poyiadjis were implemented with 50,000 particles and the $\mathcal{O}(N^2)$ algorithm was implemented with 1,000 particles. For our algorithm we set $\lambda = 0.95$ and for the fixed-lag smoother $L=7$. In terms of computational cost, given the number of particles, our algorithm has more than a 10 fold computational time saving compared to the $\mathcal{O}(N^2)$ algorithm. The fixed-lag smoother was implemented with and without the observed information matrix applied in the gradient ascent algorithm.

The RMS error of the $\mathcal{O}(N^2)$ algorithm given in Figure \ref{fig:ARpar_batch} is comparable to the error given by our $\mathcal{O}(N)$ algorithm, however, it is important to remember that this is achieved with a significant computational saving. Compared to the Poyiadjis $\mathcal{O}(N)$ algorithm, our $\mathcal{O}(N)$ algorithm and the fixed-lag smoother (using only the score estimate) produce lower RMS error. Using a fixed-lag smoother estimate of the observed information matrix in the Newton-Raphson algorithm leads to higher RMS error than when only the score is used. The poor performance of the fixed-lag approach was discussed in Section \ref{sec:comp-appr} and is attributed to the error in estimating the observed information matrix. 

Illustrating the robustness of $\lambda$ in our $\mathcal{O}(N)$ algorithm, Figure \ref{fig:PE_vary_lam} gives estimates for $\theta$ using the offline \eqref{eq:35} and online \eqref{eq:1} gradient ascent algorithms for varying values of $\lambda$ (for the online case we simulated 60,000 observations). We see that there is little difference between $\lambda=0.99$ and $\lambda=0.95$, but more importantly, for $\lambda=0.5$ the parameters are converging to the MLEs, only at a slower rate. This was also the case for much lower choices of $\lambda$ (e.g. $\lambda=0.1$), which are not shown here, but for which the parameters converged to the MLE at an even slowly rate. 
\begin{figure}[h]
  \centering
   \subfigure{\includegraphics[width=0.45\textwidth]{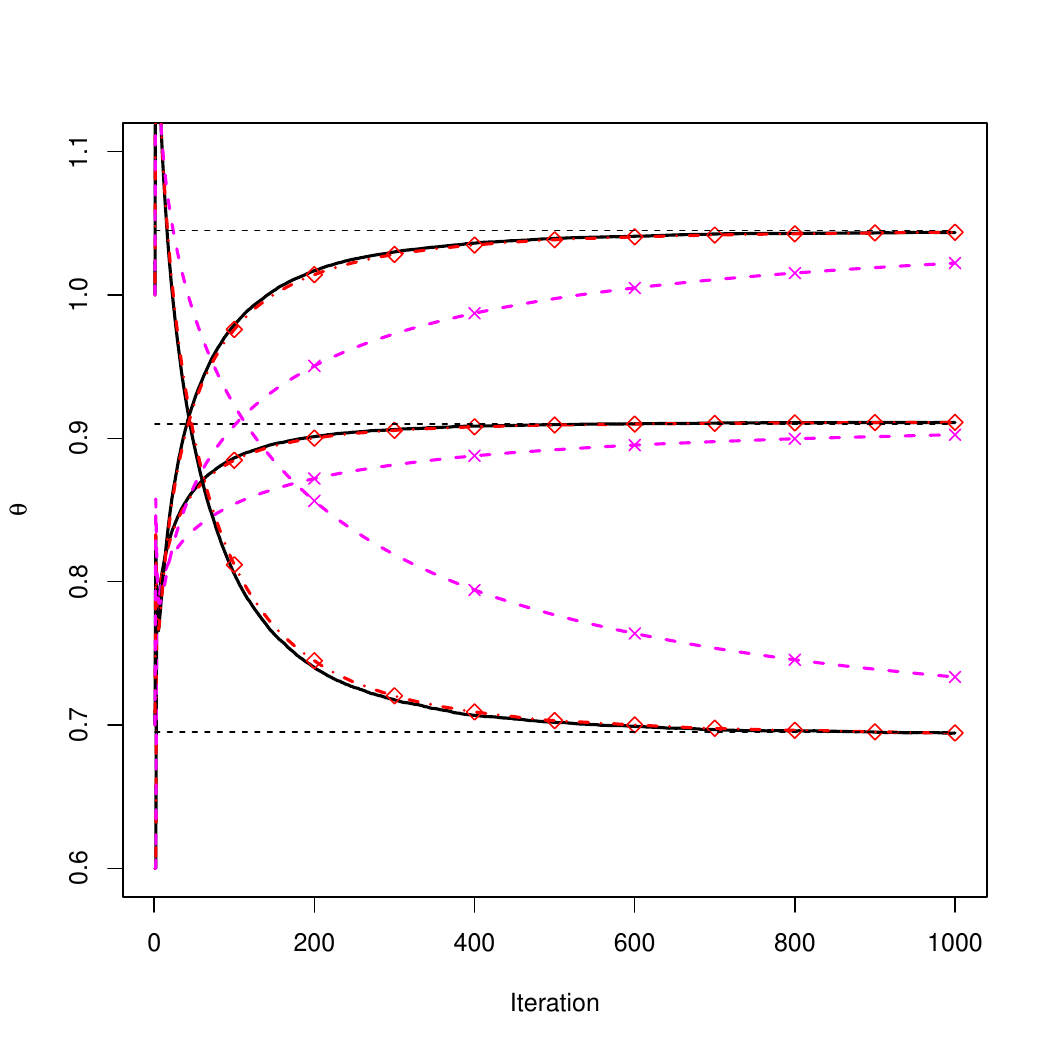}} 
   \subfigure{\includegraphics[width=0.45\textwidth]{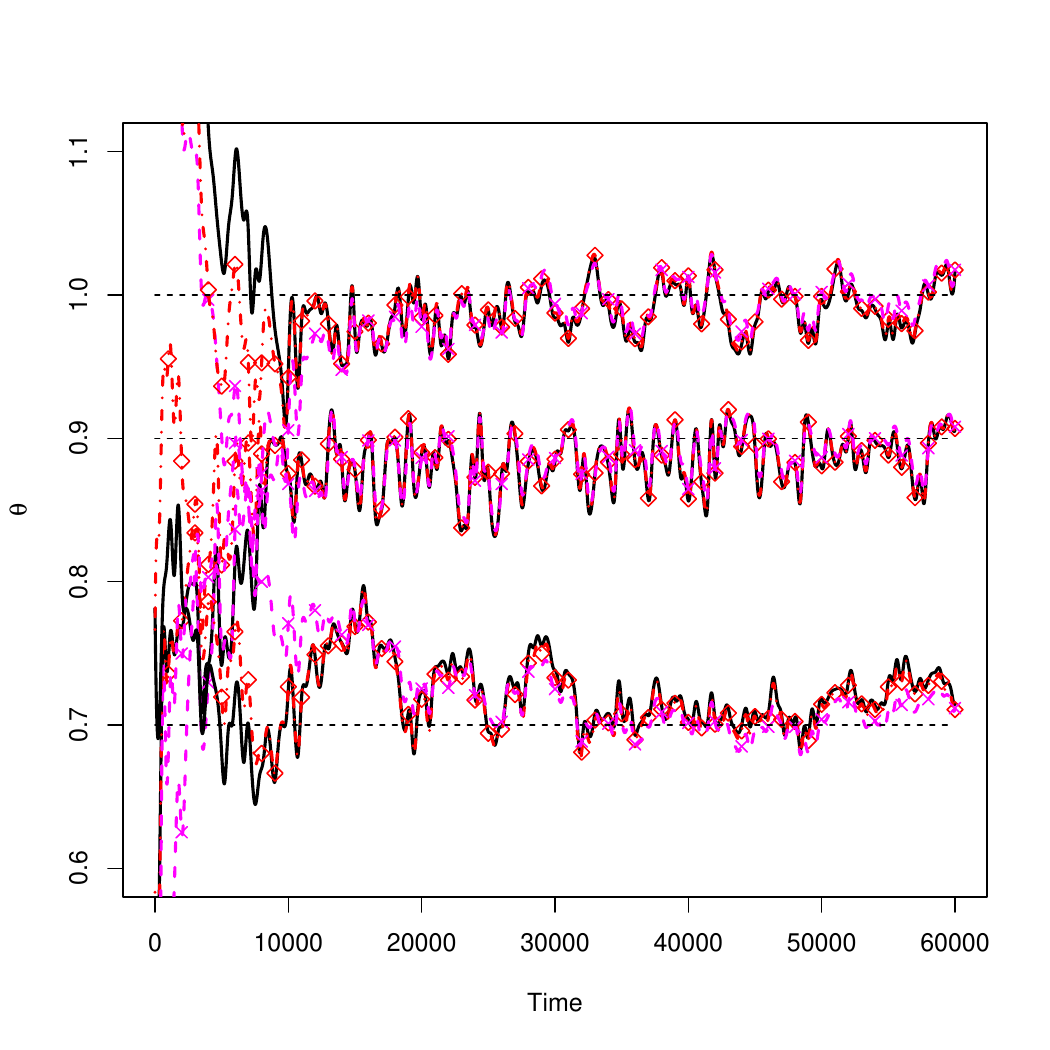}} 
   \caption{Batch (left panel) and recursive (right panel) parameter estimation for $\lambda=0.99$ (\protect\rule[0.5ex]{0.8cm}{0.6pt}), 
   $\lambda=0.95$ ( \color{red}$\boldsymbol{-} \cdot \Diamond \boldsymbol{-} \cdot\boldsymbol{-} \Diamond$  \color{black}) and $\lambda=0.5$ (  \color{cyan}$\protect\rule[0.5ex]{0.1cm}{0.5pt}$ $\times$ 
   $\protect\rule[0.5ex]{0.1cm}{0.5pt}$  $\protect\rule[0.5ex]{0.1cm}{0.5pt}$ $\times$ $\protect\rule[0.5ex]{0.1cm}{0.5pt}$  \color{black}).}
  \label{fig:PE_vary_lam}
\end{figure}

Using the recursive gradient ascent scheme \eqref{eq:1} we can compare our method against the online Bayesian particle learning algorithm \citep{Carvalho2009}. Particle learning uses MCMC moves to sequentially update the parameters within an SMC algorithm. A prior distribution is selected for each of the parameters which is updated at each time point via a set of low-dimensional sufficient statistics (see the supplementary materials for implementation details).

We generated $40,000$ observations from the AR(1) model and considered three different sets of true parameter values, chosen to represent different degrees of dependence within the underlying state process: $\phi=0.9$, 0.99 and 0.999. We set $\sigma^2 = 1-\phi^2$ so that the marginal variance of the state is $1$ and fixed $\tau=1$. We maintain the same initial parameters $\theta_0$ for the gradient scheme as was used for the batch analysis. 

\begin{figure}[h]
  \centering
   \subfigure{\includegraphics[width=0.45\textwidth]{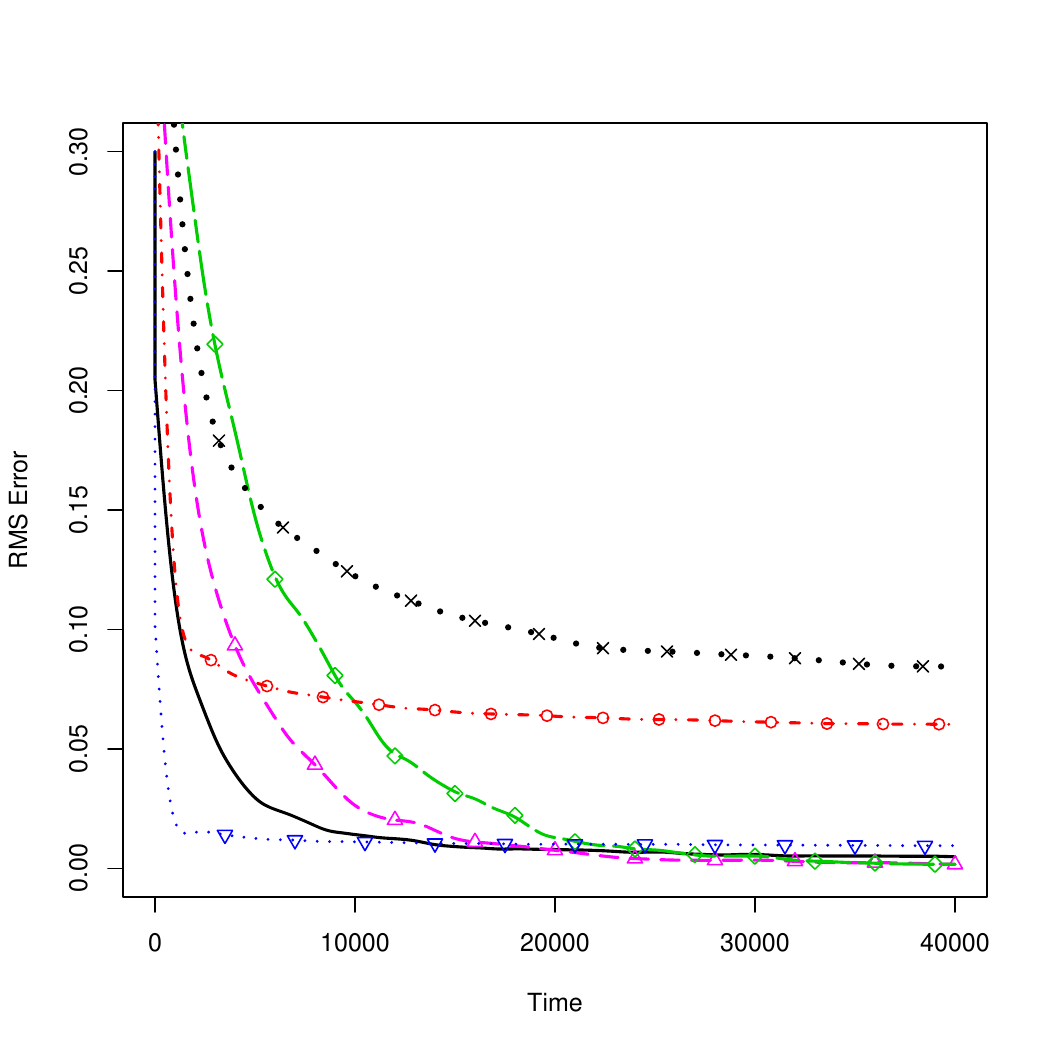}} 
   \subfigure{\includegraphics[width=0.45\textwidth]{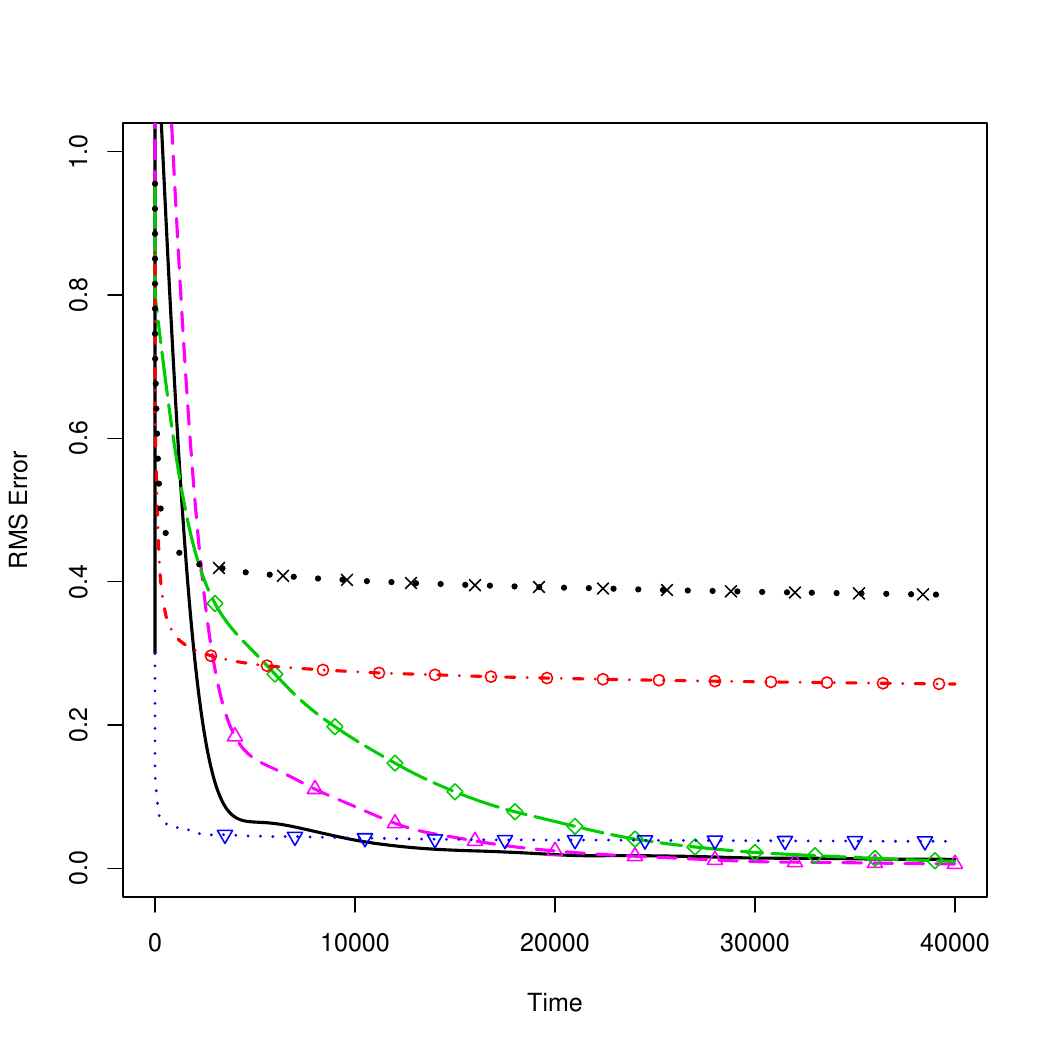}} 
   \caption{Root mean squared error of parameter estimates $\phi$ (left panel) and $\sigma$ (right panel) averaged over 100 Monte Carlo simulations from our algorithm with $\lambda=0.95$ and $\phi = 0.9$ (\protect\rule[0.5ex]{0.8cm}{0.7pt}), $\phi=0.99$ (\color{cyan} $\protect\rule[0.5ex]{0.2cm}{0.7pt}$ $\vartriangle$ $\protect\rule[0.5ex]{0.2cm}{0.7pt}$ $\protect\rule[0.5ex]{0.2cm}{0.7pt}$ $\vartriangle$ $\protect\rule[0.5ex]{0.2cm}{0.7pt}$ \color{black}), $\phi=0.999$ (\color{green} $\protect\rule[0.5ex]{0.1cm}{0.7pt}$ $\Diamond$ $\protect\rule[0.5ex]{0.4cm}{0.7pt}$ $\protect\rule[0.5ex]{0.1cm}{0.7pt}$ $\Diamond$ $\protect\rule[0.5ex]{0.4cm}{0.7pt}$ \color{black}) and the particle learning algorithm with $\phi = 0.9$ (\color{blue} $\cdot\triangledown\cdot\cdot\triangledown\cdot\cdot\cdot$ \color{black}), $\phi = 0.99$ (\color{red} $\boldsymbol{-} \circ \cdot \boldsymbol{-} \circ \cdot\boldsymbol{-}$ \color{black}), $\phi = 0.999$ ($\boldsymbol{\cdot\cdot}\times\boldsymbol{\cdot\cdot}\times\boldsymbol{\cdot\cdot}$).} 
  \label{fig:ARpar_recursive}
\end{figure}

Figure \ref{fig:ARpar_recursive} shows the RMS error of our $\mathcal{O}(N)$ algorithm applied to estimate the parameters $\theta_t$, against the particle learning filter over 100 data sets. The results show that the particle learning filter produces a lower RMS error than our algorithm for the first few thousand observations, but that it degenerates over very long time-series, particularly in the case of strong dependence ($\phi=0.99$ and 0.999). This is due to degeneracy in the sufficient statistics that occurs as a result of their dependence on the complete latent process, and the fact that the Monte Carlo approximation to $p(x_{1:T}|y_{1:T},\theta)$ degrades as $T$ increases \citep{Andrieu2005}. This degeneracy is particularly pronounced for large $\phi$, as this corresponds to cases where the underlying MCMC moves used to update the parameters mix poorly. 

Over longer data sets, applying gradient ascent with our $\mathcal{O}(N)$ algorithm, outperforms particle learning. As $\phi$ approaches 1, the long term state dependence is increased, as is the distance between the true parameter values and the fixed starting values used to initiate the gradient scheme. Our method appears to take longer to converge in this setting, but compared to particle learning, our method appears to be more robust to the choice of $\phi$, and for this reason, maximum likelihood methods are preferred over particle learning when estimating parameters from long time series. See \cite{particleLearning} for a further discussion on the implementation challenges of particle learning.

\subsection{Nonlinear seasonal Poisson model}
\label{sec:nonl-seas-poiss}

In this section we demonstrate our methodology on a nonlinear state space model, where we estimate the parameters from a real data set and show that these estimates are in agreement with previous studies. 

We consider a time series of monthly counts of poliomyelitis in the United States from January 1970 to December 1983. This time series was introduced by \cite{Zeger1988} and has since been analysed by \cite{Chan1995}, who used a Monte Carlo EM algorithm, and \cite{Davis2005} and \cite{Langrock2011} who both estimated the parameters using an approximate likelihood approach. The proposed model accounts for the observed seasonality of polio outbreaks and also contains a trend component which is the main interest in determining whether or not there is a decreasing trend:
\begin{eqnarray}
  \label{eq:13}
Y_t|X_t= x_{t},z_{t} \sim N_t[0,x_t\exp(z_t)], \quad\quad X_t|X_{t-1}=x_{t-1} \sim \mathcal{N}(\phi x_{t-1},\sigma^2) \quad\quad\quad\quad\quad\quad\\  
\log(z_t) = \mu_1+\mu_2 \frac{t}{1000} + \mu_3 \cos \left(\frac{2\pi t}{12}\right) + \mu_4 \sin \left(\frac{2\pi t}{12}\right) + \mu_5 \cos\left(\frac{2\pi t}{12}\right) + \mu_6 \sin \left( \frac{2\pi t}{12}\right), \nonumber
\end{eqnarray}
where $N_t[a,b]$ denotes the number of events in time interval $(a,b]$.

The model parameters $\theta=(\mu_1,\mu_2,\mu_3,\mu_4,\mu_5,\mu_6,\phi,\sigma^2)^\top$ are estimated using the gradient ascent algorithm, where the score vector is estimated using our proposed method (Alg. \ref{alg6}) with $\lambda=0.95$ and $0.7$.  We compare our method against the fixed-lag smoother and the Poyiadjis $\mathcal{O}(N)$ and $\mathcal{O}(N^2)$ algorithms. Each method was implemented with $N=1,000$ particles, except the Poyiadjis $\mathcal{O}(N^2)$ algorithm, which was implemented with $N=33 \approx \sqrt{1,000}$ . The fixed-lag smoother was run with lag $L=5$ and $20$.

\begin{table}[t]
  \begin{center}
\begin{tabular}{ c|cccccccc }
Algorithm & \multicolumn{8}{c}{Maximum likelihood estimates} \\
\hline
 & $\mu_1$ & $\mu_2$ & $\mu_3$ & $\mu_4$ & $\mu_5$ & $\mu_6$ & $\phi$ & $\sigma^2$ \\
Our alg. $\lambda=0.95$ &  0.26 &-3.89 & 0.16 &-0.48 & 0.41 &-0.01  &0.65  &0.28 \\
Our alg. $\lambda=0.70$ &  0.26 &-3.98 & 0.16 &-0.49 & 0.41 &-0.02  &0.61  &0.30 \\
Fixed-lag (L=5)         &  0.32 &-4.42 & 0.18 &-0.47 & 0.42 &0.00  &0.66  &0.27 \\
Fixed-lag (L=20)        &  0.32 &-4.43 & 0.18 &-0.47 & 0.42 &0.00  &0.66  &0.27 \\
Poyiadjis $\mathcal{O}(N)$ &  0.12 &-4.66 & 0.18 &-0.51  &0.41 &-0.01  &0.27  &1.00 \\
Poyiadjis $\mathcal{O}(N^2)$ & 0.21 &-3.53 & 0.14 &-0.49 & 0.43 &-0.05 & 0.66  &0.28 \\
D \& R                     & 0.24 &-3.81 & 0.16 &-0.48 & 0.41 &-0.01 & 0.63  &0.29 
\end{tabular}
\caption{\label{tab:mles} Results of batch parameter estimation for competing models using the gradient ascent algorithm \eqref{eq:35} initialised at $\theta_0=(0.4,-3,0.3,-0.3,0.65,-0.2,0.4,0.4)$. Results given by \cite{Davis2005} are quoted as D\&R.}
  \end{center}
\end{table}

Parameter estimates for the seasonal Poisson model are given in Table \ref{tab:mles}, where the batch implementation of the gradient ascent algorithm was executed for $2,000$ iterations. Given the short data set (T=168), we do not consider recursive parameter estimation.

We give the results from using our method with $\lambda=0.95$ and $\lambda=0.7$, and note that almost identical parameter estimates were obtained for $\lambda \in [0.5,0.99]$. We can see that for our method, the parameter estimates are consistent with the results presented by \cite{Davis2005} and \cite{Langrock2011}. To understand the performance of the methods we re-ran each of them 20 times to see the Monte Carlo variability in the parameter estimates. For our method, the fixed-lag smoother and the $\mathcal{O}(N^2)$ method, we obtained almost identical estimates for each run. However the $\mathcal{O}(N)$ method of Poyiadjis et al. showed increased variation in the estimates (for example the range of the estimates for $\mu_2$ was [-4.76, -4.53]). The fixed-lag smoothers performed equally well for $L=5$ and $20$ with little difference between the two implementations. Most of the parameters are estimated well using the fixed-lag smoother, but the bias of the score estimates does lead to poor estimation of $\mu_1$ and $\mu_2$. All of the algorithms, except the Poyiadjis $\mathcal{O}(N)$ and $\mathcal{O}(N^2)$ algorithms converged after approximately 500 iterations (figures available in the supplementary material). This is due to the Monte Carlo variation in the score estimates which directly impacts the parameter estimates. In the case of the $\mathcal{O}(N^2)$ algorithm, this variation could be reduced by increasing the number of particles, but at a significantly increased computational cost compared to our method.

\section{Discussion}
\label{sec:conclusions}

In this paper we have presented a novel sequential Monte Carlo method for estimating the score vector and observed information matrix for nonlinear, non-Gaussian state space models. Previous approaches have produced estimates with quadratically increasing variance at a computational cost that is linear in the number of particles, or achieved linearly increasing variance at a quadratic computational cost. 

The algorithm we have developed combines techniques from kernel density estimation and Rao-Blackwellisation to yield estimates of both the score vector and the observed information matrix which display only linearly increasing variance, which is achieved at a linear computational cost. Importantly, we have shown that this approximate score vector, at the true parameter value, has expectation zero when taken with respect to the data. Thus, the resulting gradient ascent scheme uses Monte Carlo methods to approximately find the solution to a set of unbiased estimating equations. 

The estimates of the score and observed information given by our $\mathcal{O}(N)$ algorithm can be applied to the gradient ascent and Newton-Raphson algorithms to obtain maximum likelihood estimates of the model parameters. This can be achieved either offline or online, where the parameters are estimated from a batch of observations, or recursively from observations received sequentially. Furthermore, we have shown that in terms of parameter estimation, our algorithm is relatively insensitive the the choice of $\lambda$. However we do note that setting $0.90<\lambda<0.99$ produces low variance estimates of the score with minimal bias, which also results in faster parameter convergence.

For a significant reduction in computational time we can achieve improved parameter estimation over competing methods in terms of minimising root mean squared error. We also compared our algorithm to the particle learning filter for online estimation. The particle learning filter performs well initially but degenerates over time, whereas our algorithm is more accurate over longer time series. Our method also appears to be robust to the choice of model parameters compared to the particle learning filter which struggles to estimate the parameters when the states are highly dependent.

\centering \section*{Supplementary Materials}
\label{sec:suppl-mater}

\begin{description}
\item[Appendices:] Proofs for Lemma 1 and Theorem 1. Also, a derivation of the particle learning updates and a plot for the nonseasonal Poisson model example. (pdf)

\item[R code:] R code for the examples in Section \ref{sec:parameter-estimation}. (Rcode.zip, zip file)
\end{description}

\bibliography{library,libA}
\bibliographystyle{apalike}

\end{document}